\documentclass[11pt]{article}
\usepackage{graphicx}
\usepackage{latexsym}
\usepackage{amssymb}
\usepackage{amsmath}
\usepackage{amsthm}
\usepackage{forloop}
\usepackage[paper=letterpaper,margin=1in]{geometry}
\usepackage[ruled,vlined,linesnumbered]{algorithm2e}
\usepackage{nicefrac}
\usepackage{paralist}

\newtheorem{theorem}{Theorem}
\newtheorem{lemma}{Lemma}[section]

\newtheorem{corollary}[lemma]{Corollary}

\newtheorem{observation}[lemma]{Observation}

\newtheorem{reduction}[lemma]{Reduction}

\makeatletter
\newtheorem*{rep@theorem}{\rep@title}
\newcommand{\newreptheorem}[2]{%
\newenvironment{rep#1}[1]{%
 \def\rep@title{#2 \ref{##1}}%
 \begin{rep@theorem}}%
 {\end{rep@theorem}}}
\makeatother

\newreptheorem{theorem}{Theorem}
\newreptheorem{reduction}{Reduction}

\newcommand{\defcal}[1]{\expandafter\newcommand\csname c#1\endcsname{{\mathcal{#1}}}}
\newcommand{\defbb}[1]{\expandafter\newcommand\csname b#1\endcsname{{\mathbb{#1}}}}
\newcounter{calBbCounter}
\forLoop{1}{26}{calBbCounter}{
    \edef\letter{\Alph{calBbCounter}}
    \expandafter\defcal\letter
		\expandafter\defbb\letter
}

\newcommand{\eps}{\varepsilon}
\newcommand{\ie}{{\it i.e.}}

\newcommand{\nnR}{{\bR_{\geq 0}}}

\newcommand{\SWM}{{SWM}}
\newcommand{\capend}[2]{{{#1}^{(#2)}}}
\newcommand{\etal}{{\textit{et al.}}}

\title{Online Submodular Maximization: Beating $1/2$ Made Simple}
\author{
Niv Buchbinder\thanks{Dept. of Statistics and Operations Research, Tel Aviv University, Israel. E-mail: niv.buchbinder@gmail.com}
\and
Moran Feldman\thanks{Department of Mathematics and Computer Science, The Open University of Israel. E-mail: moranfe@openu.ac.il}
\and
Yuval Filmus\thanks{Department of Computer Science, Technion, Israel. E-mail: filmus.yuval@gmail.com}
\and
Mohit Garg\thanks{Department of Mathematics and Computer Science, The Open University of Israel. E-mail: mohitga@openu.ac.il}
}

\begin{document}

\maketitle
\begin{abstract}
The \texttt{Submodular Welfare Maximization} problem (\SWM) captures an important subclass of combinatorial auctions and has been studied extensively from both computational and economic perspectives. In particular, it has been studied in a natural online setting in which items arrive one-by-one and should be allocated irrevocably upon arrival. In this setting, it is well known that the greedy algorithm achieves a competitive ratio of $\nicefrac{1}{2}$, and recently Kapralov \etal~\cite{KPV13} showed that this ratio is optimal for the problem. Surprisingly, despite this impossibility result, Korula \etal~\cite{KMZ15} were able to show that the same algorithm is $0.5052$-competitive when the items arrive in a uniformly random order, but unfortunately, their proof is very long and involved. In this work, we present an (arguably) much simpler analysis that provides a slightly better guarantee of $0.5096$-competitiveness for the greedy algorithm in the random-arrival model. Moreover, this analysis applies also to a generalization of online {\SWM} in which the sets defining a (simple) partition matroid arrive online in a uniformly random order, and we would like to maximize a monotone submodular function subject to this matroid. Furthermore, for this more general problem, we prove an upper bound of $0.576$ on the competitive ratio of the greedy algorithm, ruling out the possibility that the competitiveness of this natural algorithm matches the optimal offline approximation ratio of $1-\nicefrac{1}{e}$. 

\medskip
\noindent \textbf{Keywords:} Submodular optimization, online auctions, greedy algorithms
\end{abstract} 

\pagenumbering{Alph}
\thispagestyle{empty}
\clearpage
\pagenumbering{arabic}

\section{Introduction}

The \texttt{Submodular Welfare Maximization} problem (\SWM) captures an important subclass of combinatorial auctions and has been studied extensively from both computational and economic perspectives.
In this problem we are given a set of $m$ items and a set of $n$ bidders, where each bidder has a non-negative monotone submodular utility function,\footnote{A set function $f\colon 2^\cN \to \bR$ is \emph{monotone} if $f(S) \leq f(T)$ for every two sets $S \subseteq T \subseteq \cN$ and \emph{submodular} if $f(S \cup \{u\}) - f(S) \geq f(T \cup \{u\}) - f(T)$ for every two such sets and an element $u \in \cN \setminus T$.} and the objective is to partition the items among the bidders in a way that maximizes the total utility of the bidders. Interestingly, {\SWM} generalizes other extensively studied problems such as maximum (weighted) matching and budgeted allocation (see~\cite{M13} for a comprehensive survey).

{\SWM} is usually studied in the value oracle model (see Section~\ref{sec:preliminaries} for definition). In this model the best approximation ratio for {\SWM} is $1 - (1 - \nicefrac{1}{n})^{n} \geq (1 - \nicefrac{1}{e})$~\cite{CCPV11,FNS11,MSV08}. A different line of work studies {\SWM}  in a natural online setting in which items arrive one-by-one and should be allocated irrevocably upon arrival. This setting generalizes, for example, online (weighted) matching and budgeted allocation~\cite{AGKM11,BJN07,FKMMP09,KP00,KVV90,MSVV07,Z17}. It is well known that for this online setting the greedy approach that allocates each item to the bidder with the currently maximal marginal gain for the item is $\nicefrac{1}{2}$-competitive, which is the optimal deterministic competitive ratio~\cite{FNW78,KPV13}. While randomization is known to be very helpful for many special cases of online {\SWM} (e.g., matching), Kapralov \etal~\cite{KPV13} proved that, unfortunately, this is not the case for online {\SWM} itself---\ie, no (randomized) algorithm can achieve a competitive ratio better than $\nicefrac{1}{2}$ for this problem (unless $\mbox{NP}=\mbox{RP}$).

A common relaxation of the online setting is to assume that the items arrive in a random order rather than in an adversarial one~\cite{DH09,GM08}. This model was also studied extensively for special cases of {\SWM} for which improved algorithms were obtained~\cite{GM08,KMT11,MY11}.
Surprisingly, unlike in the adversarial setting, Korula \etal~\cite{KMZ15} showed that the simple (deterministic) greedy algorithm achieves a competitive ratio of at least $0.5052$ in the random arrival model. Unfortunately, the analysis of the greedy algorithm by Korula \etal~\cite{KMZ15} is very long and involves many tedious calculations, making it very difficult to understand why it works or how to improve it.

\subsection{Our Results}

In this paper, we study the problem of maximizing a monotone submodular function over a (simple) partition matroid. This problem is a generalization of {\SWM} (see Section~\ref{sec:preliminaries} for exact definitions and a standard reduction between the problems) in which a ground set $\cN$ is partitioned into disjoint non-empty sets $P_1, P_2, \dotsc, P_m$. The goal is to choose a subset $S\subseteq \cN$ that contains at most one element from each set $P_i$ and maximizes a given non-negative monotone submodular function $f$.\footnote{This constraint on the set of items that can be selected is equivalent to selecting an independent set of the partition matroid $\cM$ defined by the partition $\{P_1, P_2, \dotsc, P_m\}$.}
We are interested in the performance of the greedy algorithm for this problem when the sets $P_i$ are ordered uniformly at random.
A formal description of the algorithm is given as Algorithm \ref{alg:RandomOrderGreedy}.
\begin{algorithm}
\caption{\textsf{Random Order Greedy}$(f, \cM)$} \label{alg:RandomOrderGreedy}
Initialize: $A_0 \leftarrow \varnothing$.\\
Let $\pi$ be a uniformly random permutation of $[m]$.\\
\For{$i$ = $1$ \KwTo $m$}
{
    Let $u_i$ be the element $u \in P_{\pi(i)}$ maximizing $f(u \mid A_{i - 1}) \triangleq f(A_{i-1} \cup\{u\})-f(A_{i-1})$. \\
    $A_i \leftarrow A_{i - 1} \cup \{u_i\}$.
}
Return $A_m$.\\
\end{algorithm}

It is well known that for a fixed (rather than random) permutation $\pi$, the greedy algorithm achieves exactly $\nicefrac{1}{2}$-approximation~\cite{FNW78}.
We prove the following result.
\begin{theorem} \label{thm:lower_bound_ratio}
Algorithm~\ref{alg:RandomOrderGreedy} achieves an approximation ratio of at least $0.5096$ for the problem of maximizing a non-negative monotone submodular function subject to a partition matroid constraint.
\end{theorem}

Through a standard reduction from {\SWM}, this result yields the same guarantee also on the performance of the greedy algorithm for {\SWM} in the random order model.
Thus, the result both generalizes and improves over the previously known $0.5052$-approximation~\cite{KMZ15}. Our analysis is also arguably simpler, giving a direct, clean, and short proof that avoids the use of factor revealing LPs.

It should also be mentioned that the result of Korula \etal~\cite{KMZ15} represents the first combinatorial algorithm for offline {\SWM} achieving a better approximation ratio than $\nicefrac{1}{2}$. Analogously, our result is a combinatorial algorithm achieving a better than $\nicefrac{1}{2}$ approximation ratio for the more general problem of maximizing a non-negative monotone submodular function subject to a partition matroid constraint. We remark that in a recent work Buchbinder \etal~\cite{BFG18} described a (very different) offline combinatorial algorithm which achieves a better than $\nicefrac{1}{2}$ approximation for the even more general problem of maximizing a non-negative monotone submodular function subject to a general matroid constraint. However, the approximation guarantee achieved in~\cite{BFG18} is worse, and the algorithm is more complicated and cannot be implemented in an online model.

The greedy algorithm in the random arrival model is known to be $(1-\nicefrac{1}{e})$-competitive for special cases of {\SWM}~\cite{GM08}.
For online {\SWM} it is an open question whether the algorithm achieves this (best possible) ratio. However, for the more general problem of maximizing a monotone submodular function over a partition matroid, the following result answers this question negatively. In fact, the result shows that the approximation ratio obtained by the greedy algorithm is quite far from $1 - \nicefrac{1}{e} \approx 0.632$.

\begin{theorem} \label{thm:upper_bound_ratio}
There exist a partition matroid $\cM$ and a non-negative monotone submodular function $f$ over the same ground set such that the approximation ratio of Algorithm~\ref{alg:RandomOrderGreedy} for the problem of maximizing $f$ subject to the constraint defined by $\cM$ is at most $19/33 \leq 0.576$.
\end{theorem}

\subsection{Our Technique} \label{ssc:technique}

The proof we describe for Theorem~\ref{thm:lower_bound_ratio} consists of two parts. In the first part (Section~\ref{ssc:basic}), we show that when the greedy algorithm considers sets of the partition in a random order, it gains most of the value of its output set during its first iterations (Lemma~\ref{lem:intermidiate_lower_bounds}). For example, after viewing $90\%$ of the sets the algorithm already has $49.5\%$ of the value of the optimal solution, which is $99\%$ of its output guarantee according to the standard analysis. Thus, to prove that the greedy algorithm has a better than $\nicefrac{1}{2}$ approximation ratio, it suffices to show that it gets a non-negligible gain from its last iterations.

In the second part of our analysis (Section \ref{ssc:break_half}), we are able to show that this is indeed the case.
Intuitively, in this part of the analysis we view the execution of Algorithm~\ref{alg:RandomOrderGreedy} as having three stages defined by two integer values $0 < r \leq r' < m$. The first stage consists of the first $r$ iterations of the algorithm, the second stage consists of the next $r' - r$ iterations and the last stage consists of the remaining $m - r'$ iterations.
As explained above, by Lemma~\ref{lem:intermidiate_lower_bounds} we get that if $r'$ is large enough, then $f(A_{r'})$ is already very close to $f(OPT)/2$, where $OPT$ is an optimal solution.
We use two steps to prove that $f(A_m)$ is significantly larger than $f(A_{r'})$, and thus, achieves a better than $\nicefrac{1}{2}$ approximation ratio. In the first step (Lemma~\ref{lem:C_exists}), we use symmetry to argue that there are two independent sets of $\cM$ that consist only of elements that Algorithm~\ref{alg:RandomOrderGreedy} can pick in its second and third phases, and in addition, the value of their union is large. One of these sets consists of the elements of $OPT$ that are available in the final $m-r$ iterations, and the other set (which we denote by $C$) is obtained by applying an appropriately chosen function to these elements of $OPT$. In the second step of the analysis, implemented by Lemma~\ref{lem:guarantee}, we use the fact that the final $m-r'$ elements of $C$ are a random subset of $C$ to argue that they have a large marginal contribution even with respect to the final solution $A_m$. Combining this with the observation that these elements represent a possible set of elements that Algorithm~\ref{alg:RandomOrderGreedy} could pick during its last stage, we get that the algorithm must have made a significant gain during this stage.

\subsection{Additional Related Results}

The optimal approximation ratio for the problem of maximizing a monotone submodular function subject to a partition matroid constraint (and its special case {\SWM}) is obtained by an algorithm known as (Measured) Continuous Greedy~\cite{CCPV11,FNS11}. Unfortunately, this algorithm is problematic from a practical point of view since it is based on a continuous relaxation and is quite slow. As discussed above, our first result can be viewed as an alternative simple combinatorial algorithm for this problem, and thus, it is related to a line of work that aims to find better alternatives for Continuous Greedy~\cite{BV14,BFS17,FW14,MBKVK15}.

While the problem of maximizing a monotone submodular function subject to a partition matroid was studied almost exclusively in the value oracle model, the view of {\SWM} as an auction has motivated its study also in an alternative model known as the demand oracle model. In this model a strictly better than $(1-\nicefrac{1}{e})$-approximation is known for the problem~\cite{FV10}.

Another online model, that can be cast as a special case of the random arrival model and was studied extensively, is the i.i.d.\ stochastic model. In this model input items arrive i.i.d.\ according to a known or unknown distribution.
In the i.i.d.\ model with a known distribution improved competitive ratios for special cases of {\SWM} are known~\cite{AHL12,FMMM09,HMZ11,MGS12}.
Moreover, for the i.i.d.\ model with an unknown distribution a $(1-\nicefrac{1}{e})$-competitive algorithm is known for {\SWM} as well as for several of its special cases~\cite{DJSW11,DBA12,KPV13}.

\section{Preliminaries} \label{sec:preliminaries}

For every two sets $S, T \subseteq \cN$ we denote the marginal contribution of adding $T$ to $S$, with respect to a set function $f$, by $f(T \mid S) \triangleq f(T \cup S) - f(S)$. For an element $u \in \cN$ we use $f(u \mid S)$ as shorthands for $f(\{u\} \mid S)$---note that we have already used this notation previously in Algorithm~\ref{alg:RandomOrderGreedy}.

Following are two useful claims that we use in the analysis of Algorithm~\ref{alg:RandomOrderGreedy}. The first of these claims is a rephrased version of a useful lemma which was first proved in~\cite{FMV11}, and the other is a well known technical observation that we prove here for completeness.

\begin{lemma}[Lemma~2.2 of~\cite{FMV11}] \label{lem:sampling}
Let $f\colon 2^\cN \to \bR$ be a submodular function, and let $T$ be an arbitrary set $T \subseteq \cN$. For every random set $T_p \subseteq T$ which contains every element of $T$ with probability $p$ (not necessarily independently),
\[
	\bE[f(T_p)] \geq (1 - p) \cdot f(\varnothing) + p \cdot f(T)
	\enspace.
\]
\end{lemma}

\begin{observation} \label{obs:sub_mon_combine}
For every sets two $S_1 \subseteq S_2 \subseteq \cN$ and an additional set $T \subseteq \cN$, it holds that
\[
	f(S_1 \mid T) \leq f(S_2 \mid T)
	\qquad \text{and} \qquad
	f(T \mid S_1) \geq f(T \mid S_2)
	\enspace.
\]
\end{observation}
\begin{proof}
The first inequality holds since the monotonicity of $f$ implies that
\[
	f(S_1 \mid T)
	=
	f(S_1 \cup T) - f(T)
	\leq
	f(S_2 \cup T) - f(T)
	=
	f(S_2 \mid T)
	\enspace,
\]
and the second inequality holds since
\begin{align*}
	f(T \mid S_1)
	\geq{}&
	f(T \mid S_1 \cup (S_2 \setminus T)) =  f(T \cup S_2) - f(S_1 \cup (S_2 \setminus T))\\
	\geq{}&
	f(T \cup S_2) - f(S_2) = f(T \mid S_2)
	\enspace,
\end{align*}
%
where the first inequality follows from submodularity and the second from monotonicity.
\end{proof}

\paragraph{The \texttt{Submodular Welfare Maximization} problem (\SWM).} In this problem we are given a set $\cN$ of $m$ items and a set $B$ of $n$ bidders. Each bidder $i$ has a non-negative monotone submodular utility function $f_i\colon 2^\cN \to \nnR$; and the goal is to partition the items among the bidders in a way that maximizes $\sum_{i=1}^{m}f_i(S_i)$ where $S_i$ is the set of items allocated to bidder $i$.

\paragraph{Maximizing a monotone submodular function over a (simple) partition matroid.} In this problem we are given a partition matroid $\cM$ over a ground set $\cN$ and a non-negative monotone submodular function $f \colon 2^\cN \to \nnR$. A partition matroid is defined by a partition of its ground set into non-empty disjoint sets $P_1, P_2, \dotsc, P_m$. A set $S \subseteq \cN$ is independent in $\cM$ if $|S \cap P_i|\leq 1$ for every set $P_i$, and the goal in this problem is to find a set $S \subseteq \cN$ that is independent in $\cM$ and maximizes $f$. 

In this work we make the standard assumption that the objective function $f$ can be accessed only through a value oracle, \ie, an oracle that given a subset $S$ returns the value $f(S)$.

\paragraph{A standard reduction between the above two problems.} Given an instance of {\SWM}, we construct the following equivalent instance of maximizing a monotone submodular function subject to a partition matroid.
For each item $u\in \cN$ and bidder $i\in B$, we create an element $(u,i)$ which represents the assignment of $u$ to $i$.
Additionally, we define a partition of these elements by constructing for every item $u$ a set $P_u = \{(u, i) \mid i \in B\}$. 
Finally, for a subset $S$ of the elements, we define
\[
f(S)=\sum_{i \in B} f_i(\{u \in \cN \mid (u, i) \in S\})\enspace.
\]
One can verify that for every independent set $S$ the value of $f$ is equal to the total utility of the bidders given the assignment represented by $S$; and moreover, $f$ is non-negative, monotone and submodular.

It is important to note that running a greedy algorithm that inspects the partitions in a random order after this reduction is the same as running the greedy algorithm on the original {\SWM} instance in the random arrival model.

\paragraph{Additional technical reduction.}

Our analysis of Algorithm~\ref{alg:RandomOrderGreedy} uses two integer parameters $0 < r \leq r' < m$. A natural way to choose these parameters is to set them to $r= \alpha m$ and $r'= \beta m$, where $\alpha$ and $\beta$ are rational numbers. Unfortunately, not for every choice of $\alpha, \beta$ and $m$ these values are integral. The following reduction allows us to bypass this technical issue.
%

\begin{reduction} \label{red:integral}
For any fixed choice of two rational values $\alpha, \beta \in (0, 1)$, one may assume that $\alpha m$ and $\beta m$ are both integral for the purpose of analyzing the approximation ratio of Algorithm~\ref{alg:RandomOrderGreedy}.
\end{reduction}
\begin{proof}
Since $\alpha$ and $\beta$ are positive rational numbers, they can be represented as ratios $a_1/a_2$ and $b_1/b_2$, where $a_1$, $a_2$, $b_1$ and $b_2$ are all natural numbers. This implies that we can make $\alpha m$ and $\beta m$ integral by increasing $m$ by some integer value $0 \leq m' < a_2b_2$. To achieve this increase, we introduce $m'$ new dummy elements into the ground set and extend the objective function and  partition matroid in the following way. Let $D$ be the set of the $m'$ dummy elements.
\begin{compactitem}
	\item For every set $S$ that contains dummy elements, we define $f(S) = f(S \setminus D)$.
	\item For every dummy element $d \in D,$ we introduce a new set that contains only $d$ into the partition defining the matroid. Note that this implies that a set $S$ that contains dummy elements is independent if and only if $S \setminus D$ is independent.
\end{compactitem}
One can observe that this extension does not change the value of the optimal solution. Additionally, we observe that the extension does not affect the distribution of the value of the output set of Algorithm~\ref{alg:RandomOrderGreedy} because the fact that the algorithm added a dummy element to its solution does not affect either the current value of the solution or the marginals of elements considered later (in other words, the extension makes the algorithm have $m'$ new meaningless iterations in which it picks dummy elements, but it does not affect the behavior of the algorithm in the other iterations).

The above observations imply that the approximation ratio of Algorithm~\ref{alg:RandomOrderGreedy} is not affected by the extension, and thus, the approximation that the algorithm has for the extended instance (in which $\alpha m$ and $\beta m$ are integral) holds for the original instance as well.
\end{proof}

\section{Analysis of the Approximation Ratio}

In this section, we analyze Algorithm~\ref{alg:RandomOrderGreedy} and lower bound its approximation ratio. The analysis is split between Sections~\ref{ssc:basic} and~\ref{ssc:break_half}. In Section~\ref{ssc:basic} we present a basic (and quite standard) analysis of Algorithm~\ref{alg:RandomOrderGreedy} which only shows that it is a $\nicefrac{1}{2}$-approximation algorithm, but proves along the way some useful properties of the algorithm. In Section~\ref{ssc:break_half} we use these properties to present a more advanced analysis of Algorithm~\ref{alg:RandomOrderGreedy} which shows that  it is a $0.5096$-approximation algorithm (and thus proves Theorem~\ref{thm:lower_bound_ratio}).

Let us now define some notation that we use in both parts of the analysis. Let $OPT$ be an optimal solution (\ie, an independent set of $\cM$ maximizing $f$). Note that since $f$ is monotone we may assume, without loss of generality, that $OPT$ is a base of $\cM$ (\ie, it includes exactly one element of the set $P_i$ for every $1 \leq i \leq m$). Additionally, for every set $T \subseteq \cN$ we denote by $\capend{T}{i}$ the subset of $T$ that excludes elements appearing in the first $i$ sets out of $P_1, P_2, \dotsc, P_m$ when these sets are ordered according to the permutation $\pi$. More formally,
\[
	\capend{T}{i}
	=
	T \setminus \bigcup_{j = 1}^i P_{\pi(j)}
	=
	T \cap \bigcup_{j = i + 1}^m P_{\pi(j)}
	\enspace.
\]

Since $\pi$ is a uniformly random permutation and $OPT$ contains exactly one element of each set $P_i$ (due to our assumption that it is a base of $\cM$),
we get the following observation as an immediate consequence.
\begin{observation} \label{obs:opt_r_distribution}
For every $0 \leq i \leq m$, $\capend{OPT}{i}$ is a uniformly random subset of $OPT$ of size $m - i$.
\end{observation}

\subsection{Basic Analysis} \label{ssc:basic}

In this section, we present a basic analysis of Algorithm~\ref{alg:RandomOrderGreedy}. Following is the central lemma of this analysis which shows that the expression $f(A_i) + f(S \cup A_i \cup \capend{T}{i})$ is a non-decreasing function of $i$ for every pair of set $S \subseteq \cN$ and base $T$ of $\cM$ (recall that $A_i$ is the set constructed by Algorithm~\ref{alg:RandomOrderGreedy} during its $i$-th iteration). It is important to note that this lemma holds deterministically, \ie, it holds for \textbf{every} given permutation $\pi$.

\begin{lemma} \label{lem:A_OPT_development}
For every subset $S\subseteq \cN$, base $T$ of $\cM$ and $1 \leq i \leq m$,
\[
	f(A_i) + f(S \cup A_i \cup \capend{T}{i})
	\geq
	f(A_{i-1}) + f(S \cup A_{i-1} \cup \capend{T}{i-1})
	\enspace.
\]
\end{lemma}

\begin{proof}
Observe that
\begin{align*}
	f(A_i) - f(A_{i-1})
	={} &
	f(u_i \mid  A_{i-1})
	\geq
	f(T \cap P_{\pi(i)} \mid  A_{i-1})
	\geq
	f(T \cap P_{\pi(i)} \mid S \cup A_{i-1} \cup \capend{T}{i}) \\
	={} &
	f(S \cup A_{i-1} \cup \capend{T}{i-1}) - f(S \cup A_{i-1} \cup \capend{T}{i}) \\
	\geq{} &
	f(S \cup A_{i-1} \cup \capend{T}{i-1}) - f(S \cup A_{i} \cup \capend{T}{i})
	\enspace,
\end{align*}
where the first inequality follows from the greedy choice of the algorithm, the second inequality holds due to Observation~\ref{obs:sub_mon_combine} and the final inequality follows from the monotonicity of $f$.
\end{proof}

The following is an immediate corollary of the last lemma. Note that, like the lemma, it is deterministic and, thus, holds for \textbf{every} permutation $\pi$.

\begin{corollary}\label{cor:mainineq}
For every subset $S\subseteq \cN$, base $T$ of $\cM$ and $0 \leq i \leq m$,
\[
	f(A_m) + f(S \cup A_m)
	\geq
	f(A_i) + f(S \cup A_i \cup \capend{T}{i})
	\geq
	f(S \cup T)
	\enspace.
\]
\end{corollary}
\begin{proof}
Since $f(A_i) + f(S \cup A_i \cup \capend{T}{i})$ is a non-decreasing function of $i$ by Lemma~\ref{lem:A_OPT_development},
\[
	f(A_m) + f(S \cup A_m \cup \capend{T}{m})
	\geq
	f(A_i) + f(S \cup A_i \cup \capend{T}{i})
	\geq
	f(A_0) + f(S \cup A_0 \cup \capend{T}{0})
	\enspace.
\]
The corollary now follows by recalling that $A_0 = \varnothing$, observing that $f(A_0) \geq 0$ since $f$ is non-negative and noticing that by definition $\capend{T}{m} = \varnothing$ and $\capend{T}{0} = T$.
\end{proof}

By choosing $S = \varnothing$ and $T = OPT$, the last corollary yields $f(A_m) \geq \nicefrac{1}{2} \cdot f(OPT)$, which already proves that Algorithm~\ref{alg:RandomOrderGreedy} is a $\nicefrac{1}{2}$-approximation algorithm as promised. The following lemma strengthens this result by showing a lower bound on the value of $f(A_i)$ for every $0 \leq i \leq m$. Note that this lower bound, unlike the previous one, holds only in expectation over the random choice of the permutation $\pi$. Let $g(x) \triangleq x - x^2/2$.


\begin{lemma} \label{lem:intermidiate_lower_bounds}
For every $0 \leq i \leq m$, $\bE[f(A_i)] \geq g(\nicefrac{i}{m}) \cdot f(OPT)$.
\end{lemma}

\begin{proof}
As explained above, for $i = m$ the lemma follows from Corollary~\ref{cor:mainineq}. We prove the lemma for the other values of $i$ by induction. For $i = 0$ the lemma holds, even without the expectation, due to the non-negativity of $f$ since $g(0) = 0$. The rest of the proof is devoted to showing that the lemma holds for $1 \leq i < m$ assuming that it holds for $i - 1$.

Let $\pi_{i-1}$ be an arbitrary injective function from $\{1, \dotsc, i- 1\}$ to $\{1, \dotsc, m\}$, and let us denote by $\cE(\pi_{i-1})$ the event that $\pi(j) = \pi_{i-1}(j)$ for every $1 \leq j \leq i - 1$. Observe that conditioned on this event the sets $A_{i-1}$ and $\capend{OPT}{i-1}$ become deterministic. For $\capend{OPT}{i-1}$ this follows from the definition, and for $A_{i - 1}$ this is true because Algorithm~\ref{alg:RandomOrderGreedy} uses the values of $\pi$ only for the numbers in $\{1, \dotsc, i-1\}$ for constructing $A_{i - 1}$. Thus, conditioned on $\cE(\pi_{i-1})$,
\begin{align*}
	\bE[f(A_i) - f(A_{i - 1})]
	={} &
	\bE[f(u_i \mid A_{i-1})] \geq  \bE[f(OPT \cap P_{\pi(i)} \mid A_{i-1})]
	=
	\frac{\sum_{u \in \capend{OPT}{i - 1}} f(u \mid A_{i-1})}{m - i + 1}\\
	\geq{} &
	\frac{f(\capend{OPT}{i - 1} \mid A_{i-1})}{m - i + 1}
	\geq
	\frac{f(OPT) - 2f(A_{i-1})}{m - i + 1}
	\enspace,
\end{align*}
where the first inequality follows from the greedy choice of Algorithm~\ref{alg:RandomOrderGreedy}, the second equality holds since the conditioning on $\cE(\pi_{i-1})$ implies that $OPT \cap P_{\pi(i)}$ is a uniformly random element of $OPT_{i - 1}$, the second inequality follows from the submodularity of $f$ and the last inequality follows from the second inequality of Corollary~\ref{cor:mainineq} by choosing $S = \varnothing$ and $T = OPT$.

Now, taking expectation over all the possible choices of $\pi_{i-1}$, we get
\begin{align*}
	\bE[f(A_i)]
	\geq{} &
	\bE[f(A_{i-1})] + \frac{f(OPT) - 2\bE[f(A_{i-1})]}{m - i + 1}
	=
	\frac{m - i - 1}{m - i + 1} \cdot \bE[f(A_{i-1})] + \frac{f(OPT)}{m - i + 1}\\
	\geq{} &
	\frac{m - i - 1}{m - i + 1} \cdot g\left(\frac{i - 1}{m}\right) \cdot f(OPT) + \frac{f(OPT)}{m - i + 1}
	=
	\left[g\left(\frac{i - 1}{m}\right) + \frac{1 - 2g(\frac{i-1}{m})}{m - i + 1} \right] \cdot f(OPT)
	\enspace,
\end{align*}
where the second inequality follows from the induction hypothesis (since $i \leq m - 1$). Using the observations that the derivative $g'(x) = 1 - x$ of $g(x)$ is non-increasing and obeys $g'(x) = (1 - 2g(x)) / (1 - x)$, the last inequality yields
\begin{align*}
	\frac{\bE[f(A_i)]}{f(OPT)}
	\geq{} &
	g\left(\frac{i - 1}{m}\right) + \frac{1 - 2g(\frac{i-1}{m})}{m - i + 1}
	=
	g\left(\frac{i - 1}{m}\right) + \frac{g'(\frac{i-1}{m})}{m}\\
	\geq{} &
	g\left(\frac{i - 1}{m}\right) + \int_{(i - 1)/m}^{i/m} g'(x) dx
	=
	g(\nicefrac{i}{m})
	\enspace.
	\qedhere
\end{align*}
\end{proof}


\subsection{Breaking $\nicefrac{1}{2}$: An Improved Analysis of Algorithm~\ref{alg:RandomOrderGreedy}} \label{ssc:break_half}



In this section, we use the properties of Algorithm~\ref{alg:RandomOrderGreedy} proved in the previous section to derive a better than $\nicefrac{1}{2}$ lower bound on its approximation ratio and prove Theorem~\ref{thm:lower_bound_ratio}.
As explained in Section~\ref{ssc:technique}, we view here an execution of Algorithm~\ref{alg:RandomOrderGreedy} as consisting of three stages, where the places of transition between the stages are defined by two integer parameters $0 < r \leq r' < m$ whose values are chosen later in this section.
The first lemma that we present (Lemma~\ref{lem:C_exists}) uses a symmetry argument to prove  that there are two (not necessarily distinct) independent sets of $\cM$ that consist only of elements that Algorithm~\ref{alg:RandomOrderGreedy} can pick in its second and third stages (the final $m-r$ iterations), and in addition, the value of their union is large. One of these sets is $\capend{OPT}{r}$, and the other set is obtained by applying to $\capend{OPT}{r}$ an appropriately chosen function $h$.

Let $c$ be the true (unknown) approximation ratio of Algorithm~\ref{alg:RandomOrderGreedy}.

\begin{lemma} \label{lem:C_exists}
There exists a function $h\colon 2^\cN \to 2^\cN$ such that
\begin{enumerate}[(a)]
	\item for every $1 \leq i \leq m$ and set $S \subseteq \cN$, $|P_i \cap h(S)| = |P_i \cap S|$. \label{item:feasibility}
	\item $\bE[f(h(\capend{OPT}{r}) \cup \capend{OPT}{r})] \geq f(OPT) - c^{-1} \cdot \bE[f(A_m \mid A_{m-r})]$. \label{item:value}
\end{enumerate}
\end{lemma}

\begin{proof}
Given Part~\eqref{item:value} of the lemma, it is natural to define $h(S)$, for every set $S \subseteq \cN$, as the set $T$ maximizing $f(T \cup S)$ among all the sets obeying Part~\eqref{item:feasibility} of the lemma (where ties are broken in an arbitrary way). In the rest of the proof we show that this function indeed obeys Part~\eqref{item:value}.

Observe that
\begin{align*}
	f(A_{m - r} \cup{} & (OPT \setminus \capend{OPT}{m-r}))
	=
	f(OPT \setminus \capend{OPT}{m-r} \mid A_{m - r}) +f(A_{m-r}) \\
	&\geq
	f(OPT \setminus \capend{OPT}{m-r} \mid A_{m - r} \cup \capend{OPT}{m-r}) + f(A_{m-r})\\
	&=
	f(A_{m - r} \cup OPT) - f(\capend{OPT}{m-r} \mid A_{m - r})
	\geq
	f(OPT) - f(\capend{OPT}{m-r} \mid A_{m - r})
	\enspace,
\end{align*}
where first inequality follows from Observation~\ref{obs:sub_mon_combine} and the second follows by the monotonicity of $f$.
We now note that the last $r$ iterations of Algorithm~\ref{alg:RandomOrderGreedy} can be viewed as a standalone execution of this algorithm on the partition matroid defined by the sets $P_{m - r + 1}, \dotsc, P_m$ and the objective function $f(\cdot \mid A_{m-r})$. Thus, by the definition of $c$, the expected value of $f(\capend{OPT}{m-r} \mid A_{m - r})$ is at most $c^{-1} \cdot \bE[f(A_m \setminus A_{m-r} \mid A_{m-r})] = c^{-1} \cdot \bE[f(A_m \mid A_{m-r})]$. Combining this with the previous inequality, we get
\begin{align*}
	\bE[f(h(OPT \setminus{} & \capend{OPT}{m-r}) \cup (OPT \setminus \capend{OPT}{m-r}))]
	\geq
	\bE[f(A_{m-r} \cup (OPT \setminus \capend{OPT}{m-r}))]\\
	\geq{} &
	\bE[f(OPT) - f(\capend{OPT}{m-r} \mid A_{m - r})]
	\geq
	f(OPT) - c^{-1} \cdot \bE[f(A_m \mid A_{m-r})]
	\enspace,
\end{align*}
where the first inequality holds due to the definition of $h$ since $A_{m-r}$ obeys Part~\eqref{item:feasibility} of the lemma (for $S = OPT \setminus \capend{OPT}{m-r}$).

To prove the lemma it remains to observe that by Observation~\ref{obs:opt_r_distribution} the random sets $OPT \setminus \capend{OPT}{m-r}$ and $\capend{OPT}{r}$ have the same distribution, which implies that $f(h(OPT\setminus \capend{OPT}{m-r}) \cup (OPT \setminus \capend{OPT}{m-r}))$ and $f(h(\capend{OPT}{r}) \cup \capend{OPT}{r})$ have the same expectation.
\end{proof}

Let us denote $C = h(\capend{OPT}{r})$. Note that $C$ is a random set since $\capend{OPT}{r}$ is. The following lemma uses the properties of $C$ proved by Lemma~\ref{lem:C_exists} to show that Algorithm~\ref{alg:RandomOrderGreedy} must make a significant gain during its third stage.

\begin{lemma}\label{lem:guarantee}
Let $p = \frac{m - r'}{m - r}$, then
\[
	\bE[f(A_m \mid A_{r'})]
	\geq
	f(OPT) - (2+p/c) \cdot \bE[f(A_m)] + p \cdot \bE[f(A_r)] + (p/c) \cdot \bE[f(A_{m-r})]
	\enspace.
\]
\end{lemma}

\begin{proof}
Observe that
\begin{align*}
	f(A_m \mid A_{r'})
	\geq{} &
	f(\capend{C}{r'} \mid A_m \cup \capend{OPT}{r})
	=
	f(A_m \cup \capend{C}{r'} \mid A_r \cup \capend{OPT}{r}) - f(A_m \mid A_r \cup \capend{OPT}{r}) \\
	\geq{} &
	f(\capend{C}{r'} \mid A_r \cup \capend{OPT}{r}) - f(A_m \mid A_r \cup \capend{OPT}{r})
	\enspace,
\end{align*}
where the first inequality follows by plugging $i = r'$, $T = A_r \cup C$ and $S = A_m \cup \capend{OPT}{r}$ in the first inequality of Corollary~\ref{cor:mainineq}, and the second inequality follows by Observation~\ref{obs:sub_mon_combine}.

Similar to what we do in the proof of Lemma~\ref{lem:intermidiate_lower_bounds}, let us now denote by $\pi_r$ an arbitrary injective function from $\{1, \dotsc, r\}$ to $\{1, \dotsc, m\}$ and by $\cE(\pi_r)$ the event that $\pi(j) = \pi_r(j)$ for every $1 \leq j \leq r$. Observe that conditioned on $\cE(\pi_r)$ the set $\capend{OPT}{r}$ is deterministic, and thus so is the set $C$ which is obtained from $\capend{OPT}{r}$ by the application of a deterministic function; but $\capend{C}{r'}$ remains a random set that contains every element of $C$ with probability $p$. Hence, by Lemma~\ref{lem:sampling}, conditioned on $\cE(\pi_r)$, we get
\[
	\bE[f(\capend{C}{r'} \mid A_r \cup \capend{OPT}{r})]
	\geq
	p \cdot f(C \mid A_r \cup \capend{OPT}{r})
	\enspace.
\]
Taking now expectation over all the possible events $\cE(\pi_r)$, and combining with the previous inequality, we get
\begin{align*}
	\bE[f(A_m \mid A_{r'})]
	&\geq
	p \cdot \bE[f(C \mid A_r \cup \capend{OPT}{r})] - \bE[f(A_m \mid A_r \cup \capend{OPT}{r})]\\
	={} &
	p \cdot \bE[f(C \cup A_r \cup \capend{OPT}{r})] + (1 - p) \cdot \bE[f(A_r \cup \capend{OPT}{r})] - \bE[f(A_m \cup \capend{OPT}{r})]\\
	\geq{} &
	p \cdot \bE[f(C \cup \capend{OPT}{r})] + (1 - p) \cdot \bE[f(A_r \cup \capend{OPT}{r})] - \bE[f(A_m \cup \capend{OPT}{r})]
	\enspace,
\end{align*}
where the second inequality holds due to the monotonicity of $f$. We now need to bound all the terms on the right hand side of the last inequality. The first term is lower bounded by $p \cdot f(OPT) - (p/c) \cdot \bE[f(A_m \mid A_{m-r})]$ due to Lemma~\ref{lem:C_exists}. A lower bound of $f(OPT) - f(A_r)$ on the expression $f(A_r \cup \capend{OPT}{r})$ follows from the second inequality of Corollary~\ref{cor:mainineq} by setting $T = OPT$ and $S = \varnothing$. Finally, an upper bound of $2f(A_m) - f(A_r)$ on the expression $f(A_m \cup \capend{OPT}{r})$ follows from the first inequality of the same corollary by setting $T = OPT$ and $S = A_m$. Plugging all these bounds into the previous inequality yields
\begin{align*}
	\bE[f(&A_m \mid A_{r'})]\\
	\geq{} &
	p \cdot f(OPT) - (p/c) \cdot \bE[f(A_m \mid A_{m-r})] + (1 - p) \cdot \bE[f(OPT) - f(A_r)] - \bE[2f(A_m) - f(A_r)]\\
	={} &
	f(OPT) - (2 + p/c) \cdot \bE[f(A_m)] + p \cdot \bE[f(A_r)] + (p/c) \cdot \bE[f(A_{m-r})]
	\enspace.
	\qedhere
\end{align*}
\end{proof}

We are now ready to prove Theorem~\ref{thm:lower_bound_ratio}.


\begin{proof}[Proof of Theorem \ref{thm:lower_bound_ratio}]
Let $q = r/m$. By Lemma~\ref{lem:guarantee},
\begin{align*}
	\bE[f(A_m)]
	={} &
	\bE[f(A_m \mid A_{r'})] + \bE[f(A_{r'})]\\
	\geq{} &
	f(OPT) - (2+p/c) \cdot \bE[f(A_m)] + p \cdot \bE[f(A_r)] + (p/c) \cdot \bE[f(A_{m-r})]+ \bE[f(A_{r'})]
	\enspace.
\end{align*}
Rearranging this inequality, and using the lower bound on $\bE[f(A_i)]$ given by Lemma~\ref{lem:intermidiate_lower_bounds}, we get
\begin{align*}
	(3 + p/c) \cdot \bE[f(A_m)]
	\geq{} &
	[1 + p \cdot g(q) + (p/c) \cdot g(1 - q) + g(1 - p + pq)] \cdot f(OPT)\\
	={} &
	[1 + pq(1 - q/2) + (p/c)(1 - q^2)/2 + (1 - p^2 + 2p^2q - p^2q^2)/2] \cdot f(OPT)\\
	={}&
	\frac{1}{2}[3 + pq(2 - q) + pc^{-1}(1 - q^2) - p^2(1 - q)^2] \cdot f(OPT)
	\enspace.
\end{align*}
Thus, the approximation ratio of Algorithm~\ref{alg:RandomOrderGreedy} is at least
\[
	\frac{3 + pq(2 - q) + pc^{-1}(1 - q^2) - p^2(1 - q)^2}{6 + 2pc^{-1}}
	\enspace.
\]

Since $c$ is the true approximation ratio of this algorithm by definition, we get
\[
	c \geq \frac{3 + pq(2 - q) + pc^{-1}(1 - q^2) - p^2(1 - q)^2}{6 + 2pc^{-1}}
	\enspace.
\]
We now choose $p = q = 0.4$. Notice that these values for $p$ and $q$ can be achieved by setting $r = \alpha m$ and $r' = \beta m$ for an appropriate choice of $0 < \alpha < \beta < 1$, and moreover, we can assume that this is a valid choice for $r$ and $r'$ by Reduction~\ref{red:integral}. Plugging these values of $p$ and $q$ into the last inequality and simplifying, we get $6c^2 - 2.3984c - 0.336 \geq 0$.
One can verify that all the positive solutions for this inequality are larger than $0.5096$, which completes the proof of the theorem.
%
%
%
%
\end{proof}

\section{Upper Bounding the Approximation Ratio} \label{sec:upper-bound}

In this section, we prove a weaker form of Theorem~\ref{thm:upper_bound_ratio}, with a bound of $7/12 \approx 0.583$ instead of $19/33 \approx 0.576$. The proof of the theorem as stated appears in the appendix. 

	To prove the weaker version of the theorem, we construct a partition matroid over a ground set $\cN$ consisting of twelve elements and a non-negative monotone submodular function $f\colon 2^\cN \to \nnR$ over the same ground set. The partition matroid is defined by a partition of the ground set into three sets: $P_x=\{x_1, x_2, x_3, x_4\}$, $P_y=\{y_1, y_2, y_3, y_4\}$ and $P_z=\{z_1, z_2, z_3, z_4\}$.
	To define the function $f$, we view each element of $\cN$ as a subset of an underlying universe $\cU$ consisting of $12$ elements: \[\cU=\{\alpha_1,\cdots,\alpha_4, \beta_1, \cdots, \beta_4, \gamma_1, \cdots, \gamma_4\}.\] The function $f$ is then given as the coverage function $f(S)=|\bigcup_{u\in S} u|$ (coverage functions are known to be non-negative, monotone and submodular). The following table completes the definition of $f$ by specifying the exact subset of $\cU$ represented by each element of $\cN$:

	\begin{center}
		{\footnotesize
		\begin{tabular}{ l | l | l }
			\multicolumn{1}{c|}{Elements of $P_x$} & \multicolumn{1}{c|}{Elements of $P_y$} & \multicolumn{1}{c}{Elements of $P_z$}\\
			\hline&&\\[-2.8mm]
			$x_1=\{\alpha_1,\alpha_2,\alpha_3,\alpha_4\}$ &
			$y_1=\{\beta_1,\beta_2,\beta_3,\beta_4\}$ &
			$z_1=\{\gamma_1,\gamma_2,\gamma_3,\gamma_4\}$ \\
			$x_2=\{\beta_1,\beta_2,\gamma_1,\gamma_2\}$ & 
			$y_2=\{\alpha_1,\alpha_2,\gamma_1,\gamma_2\}$ &
			$z_2=\{\alpha_1,\alpha_2,\beta_1,\beta_2\}$ \\
			$x_3=\{\beta_1,\gamma_3\}$ &
			$y_3=\{\alpha_1,\gamma_3\}$ &
			$z_3=\{\alpha_1,\beta_3\}$ \\
			$x_4=\{\beta_3,\gamma_1\}$ &
			$y_4=\{\alpha_3,\gamma_1\}$ &
			$z_4=\{\alpha_3,\beta_1\}$
		\end{tabular}
		}
	\end{center}
	
	It is easy to verify that the optimum solution for this instance (\ie, the independent set of $\cM$ maximizing $f$) is the set $\{x_1,y_1,z_1\}$ whose value is $12$. To analyze the performance of Algorithm~\ref{alg:RandomOrderGreedy} on this instance, we must set a tie breaking rule. Here we assume that the algorithm always breaks ties in favor of the element with the higher index, but it should be noted that a small perturbation of the values of $f$ can be used to make the analysis independent of the tie breaking rule used (at the cost of weakening the impossibility proved by an additive $\eps$ term for an arbitrary small constant $\eps > 0$).
	
Now, consider the case that the set $P_x$ arrives first, followed by $P_y$ and finally $P_z$. One can check that in this case the greedy algorithm picks the elements $x_2$, $y_3$ and $z_4$ (in this order), and that their marginal contributions upon arrival are $4$, $2$ and $1$, respectively. Similarly, it can be checked that the exact same marginal contributions also appear in every one of the other five possible arrival orders of the sets $P_x, P_y, P_z$. 
Thus, regardless of the arrival order, the approximation ratio achieved by Algorithm~\ref{alg:RandomOrderGreedy} for the above instance is only
\[
	\frac{4+2+1}{12}
	=
	\frac{7}{12}
	\enspace.
\]

\noindent \textbf{Remark:} It should be noted that by combining multiple independent copies of the above described instance, one can get an arbitrarily large instance for which the approximation ratio of Algorithm~\ref{alg:RandomOrderGreedy} is only $7/12$. This rules out the possibility that the approximation ratio of Algorithm~\ref{alg:RandomOrderGreedy} approaches $1 - \nicefrac{1}{e}$ for large enough instances.

%
%
\section*{Acknowledgment}

We thank Nitish Korula, Vahab S. Mirrokni and Morteza Zadimoghaddam for sharing with us the full version of their paper~\cite{KMZ15}.

\bibliographystyle{plain}
\bibliography{SubmodularMax}

\newpage
\appendix
\section{Stronger Upper Bound on the Approximation Ratio}

In this appendix we prove Theorem~\ref{thm:upper_bound_ratio} with the stated bound of $19/33$. The proof is similar to the one given in Section~\ref{sec:upper-bound}, but the set system we need to use is more complicated.

We construct a partition matroid over a ground set $\cN$ consisting of $32$ elements and a non-negative monotone submodular function $f\colon 2^\cN \to \nnR$ over the same ground set. The ground set of the partition matroid consists of four types of elements:
\begin{itemize}
	\item $o_1, o_2, o_3, o_4$
	\item $x_1, x_2, x_3, x_4$
	\item $y_{ij}$ for distinct $i,j \in \{1,2,3,4\}$
	\item $z_{ijk}$	for distinct $i,j,k \in \{1,2,3,4\}$, where we ignore the order between the first two indices (\ie, $z_{ijk}$ and $z_{jik}$ are two names for the same element)
\end{itemize}
The partition matroid is defined by a partition of the ground set into four parts $P_1,P_2,P_3,P_4$, where part $P_i$ is comprised of the $8$ elements of the forms $o_i, x_i, y_{ji}, z_{jki}$.

As in Section~\ref{sec:upper-bound}, the function $f$ is a coverage function, but this time a weighted one. The universe $\cU$ used to define this function consists of $28$ elements:
\[
 \cU = \{ a_i, b_i, c_i, d_i, e_i, f_i, g_i : i \in \{1,2,3,4\} \}.
\]
The weights of the elements in this universe are given by the following function $w\colon \cU \to \nnR$:
\[
\begin{array}{c|ccccccc}
	v & a_i & b_i & c_i & d_i & e_i & f_i & g_i \\\hline
	w(v) & 14 & 14 & 8 & 5 & 4 & 7 & 14
\end{array}
\]
We extend $w$ to subsets of $\cU$ by defining $w(V) = \sum_{v \in V} w(v)$.
The function $f$ is then the weighted coverage function given by the formula $f(S) = w(\bigcup_{v \in S} v)$. Like coverage functions, weighted coverage functions (with non-negative weights) are also known to be non-negative, monotone and submodular.

In order to complete the definition of $f$, we need to specify the sets that the elements of $\cN$ correspond to:
\begin{align*}
o_i &= \{a_i,b_i,c_i,d_i,e_i,f_i,g_i\} \\
x_i &= \{b_j,c_j : j \neq i\} \\
y_{ij} &= \{c_i, e_j\} \cup \{d_k,e_k,f_k : k \neq i,j\} \\	
z_{ijk} &= \{f_i,f_j,g_\ell\}, \text{ where } \ell \text{ is the unique element of } \{1,2,3,4\} \setminus \{i,j,k\}
\end{align*}

It is easy to verify that the optimal solution for this instance is $\{o_1,o_2,o_3,o_4\}$ and that it achieves a total weight of $264$. In contrast, by choosing an appropriate tie breaking rule, we can cause Algorithm~\ref{alg:RandomOrderGreedy} to act in the following way when presented with the parts $P_i,P_j,P_k,P_\ell$: choose $x_i,y_{ij},z_{ijk},o_\ell$ (which has total weight 152). In order to demonstrate this, we analyze below the working of the algorithm when the parts are presented in the order $P_1,P_2,P_3,P_4$. But first, let us notice that
\[
 x_1 \cup y_{12} \cup z_{123} \cup o_4 =
 \{ c_1,f_1,b_2,c_2,e_2,f_2,b_3,c_3,d_3,e_3,f_3,a_4,b_4,c_4,d_4,e_4,f_4,g_4 \}
\]
has total weight $152$, which is smaller than the optimum by a factor of $152/264 = 19/33$.

\pagebreak

At the first step, the algorithm has the following options (when several options are equivalent, only one is presented):
\begin{enumerate}
	\item $o_1 = \{a_1,b_1,c_1,d_1,e_1,f_1,g_1\}$, total weight $66$.
	\item $x_1 = \{b_2,c_2,b_3,c_3,b_4,c_4\}$, total weight $66$.
	\item $y_{21} = \{c_2,e_1,d_3,e_3,f_3,d_4,e_4,f_4\}$, total weight $44$.
	\item $z_{321} = \{f_3,f_2,g_4\}$, total weight $28$.
\end{enumerate}
Therefore $x_1$ is a legitimate choice.
At the second step the options are:
\begin{enumerate}
	\item $o_2 \setminus x_1 = \{a_2,d_2,e_2,f_2,g_2\}$, total weight $44$.
	\item $x_2 \setminus x_1 = \{b_1,c_1\}$, total weight $22$.
	\item $y_{12} \setminus x_1 = \{c_1,e_2,d_3,e_3,f_3,d_4,e_4,f_4\}$, total weight $44$.
	\item $y_{32} \setminus x_1 = \{e_2,d_1,e_1,f_1,d_4,e_4,f_4\}$, total weight $36$.
	\item $z_{132} \setminus x_1 = \{f_1,f_3,g_4\}$, total weight $28$.
	\item $z_{342} \setminus x_1 = \{f_3,f_4,g_1\}$, total weight $28$.
\end{enumerate}
Therefore $y_{12}$ is a legitimate choice.
At the third step the options are:
\begin{enumerate}
	\item $o_3 \setminus (x_1 \cup y_{12}) = \{a_3,g_3\}$, total weight $28$.
	\item $x_3 \setminus (x_1 \cup y_{12}) = \{b_1\}$, total weight $14$.
	\item $y_{13} \setminus (x_1 \cup y_{12}) = \{d_2,f_2\}$, total weight $12$.
	\item $y_{23} \setminus (x_1 \cup y_{12}) = \{d_1,e_1,f_1\}$, total weight $16$.
	\item $y_{43} \setminus (x_1 \cup y_{12}) = \{d_1,e_1,f_1,d_2,f_2\}$, total weight $28$.
	\item $z_{123} \setminus (x_1 \cup y_{12}) = \{f_1,f_2,g_4\}$, total weight $28$.
	\item $z_{143} \setminus (x_1 \cup y_{12}) = \{f_1,g_2\}$, total weight $21$.
	\item $z_{243} \setminus (x_1 \cup y_{12}) = \{f_2,g_1\}$, total weight $21$.
\end{enumerate}
Therefore $z_{123}$ is a legitimate choice.
At the final step the options are:
\begin{enumerate}
	\item $o_4 \setminus (x_1 \cup y_{12} \cup z_{123}) = \{a_4\}$, total weight $14$.
	\item $x_4 \setminus (x_1 \cup y_{12} \cup z_{123}) = \{b_1\}$, total weight $14$.
	\item $y_{14} \setminus (x_1 \cup y_{12} \cup z_{123}) = \{d_2\}$, total weight $5$.
	\item $y_{24} \setminus (x_1 \cup y_{12} \cup z_{123}) = \{d_1,e_1\}$, total weight $9$.
	\item $y_{34} \setminus (x_1 \cup y_{12} \cup z_{123}) = \{d_1,e_1,d_2\}$, total weight $14$.
	\item $z_{124} \setminus (x_1 \cup y_{12} \cup z_{123}) = \{g_3\}$, total weight $14$.
	\item $z_{134} \setminus (x_1 \cup y_{12} \cup z_{123}) = \{g_2\}$, total weight $14$.
	\item $z_{234} \setminus (x_1 \cup y_{12} \cup z_{123}) = \{g_1\}$, total weight $14$.
\end{enumerate}
We see that all options increase the value by at most $14$.

\end{document}